\documentclass{article}

\usepackage{arxiv}

\usepackage[utf8]{inputenc} 
\usepackage[T1]{fontenc}    
\usepackage{fancyhdr}
\usepackage{lmodern}
\usepackage[colorlinks]{hyperref}  
\usepackage{url}            
\usepackage{booktabs}       
\usepackage{amsfonts}       
\usepackage{nicefrac}       
\usepackage{microtype}      
\usepackage{lipsum}

\usepackage{graphicx}
\usepackage{epsfig}
\usepackage{epstopdf}
\usepackage{subfigure}

\usepackage{enumerate}
\usepackage{amssymb}
\usepackage{amsmath}
\usepackage{amsthm}
\usepackage{enumitem}
\usepackage{mathrsfs}

\newcommand{\R}{\mathbb{R}}

\newcommand{\N}{\mathbb{N}}

\newcommand{\Ra}{\mathcal{R}} 

\renewcommand{\rho}{\varrho}
\renewcommand{\phi}{\varphi}
\renewcommand{\b}{\widehat{\beta}}

\newcommand{\scal}[2]{\left\langle#1,#2\right\rangle}

\newtheorem{rem}{Remark}[section]

\newtheorem{thm}{Theorem}[section]

\title{Compressibility effect on Darcy porous convection}

\author{  
G. Arnone$^1$\href{https://orcid.org/0000-0002-3317-6358}{\includegraphics[scale=0.1]{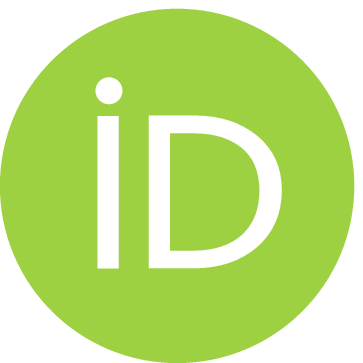}} \\  
\texttt{giuseppe.arnone@unina.it} \\ 
\And
F. Capone$^1$\thanks{Corresponding author.} \href{https://orcid.org/0000-0002-0672-999X}{\includegraphics[scale=0.1]{orcid.eps}} \\
\texttt{fcapone@unina.it} \\
\And 
R. De Luca$^1$\href{https://orcid.org/0000-0002-2109-7564}{\includegraphics[scale=0.1]{orcid.eps}}   \\
\texttt{roberta.deluca@unina.it} \\     
\And  
G. Massa$^1$\href{https://orcid.org/0000-0002-8401-9176}{\includegraphics[scale=0.1]{orcid.eps}} \\
\texttt{giuliana.massa@unina.it} \\ \\ $^1$Dipartimento di Matematica e Applicazioni 'R.Caccioppoli' \\ Università degli Studi di Napoli Federico II \\ Via Cintia, Monte S.Angelo, 80126 Napoli \\ Italy \\  }

\renewcommand{\shorttitle}{Compressibility effect on Darcy porous convection}

\begin{document}
\maketitle

\begin{abstract}
Perfectly incompressible materials do not exist in nature but are a useful approximation of several media which can be deformed in non-isothermal processes but undergo very small volume variation. In this paper the linear analysis of the Darcy-Bénard problem is performed in the class of extended-quasi-thermal-incompressible fluids, introducing a factor $\beta$ which describes the compressibility of the fluid and plays an essential role in the instability results. In particular, in the Oberbeck-Boussinesq approximation, a more realistic constitutive equation for the fluid density is employed in order to obtain more thermodynamic consistent instability results. Via linear instability analysis of the conduction solution, the critical Rayleigh-Darcy number for the onset of convection is determined as a function of a dimensionless parameter $\widehat{\beta}$ proportional to the compressibility factor $\beta$, proving that $\widehat{\beta}$ enhances the onset of convective motions. 
\end{abstract}
\keywords{Porous Media \and Incompressible fluids \and Boussinesq approximation \and Compressibility effect \and Instability analysis \and Extended-Quasi-Thermal-Incompressible fluids}

\section{Introduction}

The mathematical models describing the onset of convective motions in horizontal layers of fluids heated from below are well known for both clear fluids and fluid-saturated porous media (see \cite{chandr, NieldBejan2017} and references therein) and have been widely analysed under various assumptions: in \cite{jacopo1,jacopo2,jacopo3} the authors analysed the effect of the local thermal non-equilibrium hypothesis on the onset of convection in horizontal porous layers; in \cite{rees2020} the Darcy-Bénard problem for Bingham fluids has been studied, while in \cite{barletta2012,celli2019} the authors examined the onset of convection in an inclined horizontal layer of porous medium; convective instabilities in horizontal layers of bi-disperse porous media have been analysed in \cite{giuliana1,giuliana2,challoob2021,giuliana3,giuliana4}, the onset of penetrative convection has been studied in \cite{CapGentHill2011,giuseppe1,CapGentHill2010}. Usually, the fluid is assumed as Newtonian and incompressible. However, this is an approximation of the real phenomenon, since perfectly incompressible fluids do not exist in nature and, moreover, when the process is not isothermal, the notion of incompressibility is not well defined, see \cite{Gouinrogers2012}. From a mathematical point of view, the pressure for a compressible fluid is a constitutive function, while the pressure for an incompressible fluid is a Lagrange multiplier that comes from the constraint of incompressibility. To study and compare the mathematical results and solutions of both compressible and incompressible media, we will consider the pressure $p$ and the temperature $T$ as thermodynamic variables, therefore $V=V(p,T)$ and $\varepsilon=\varepsilon(p,T)$ are the constitutive equations for the specific volume $V=\frac{1}{\rho}$ ($\rho$ being the fluid density) and the internal energy of the system $\varepsilon$ \cite{ruggeri2021classical}. \\
According to M\"uller, see \cite{Muller1985}, an incompressible fluid can be defined as a medium whose constitutive equations depend only on temperature $T$ and not on pressure $p$, in particular:
\begin{equation}\label{muller}
	\rho=\rho(T), \quad \varepsilon=\varepsilon(T)
\end{equation}
Nevertheless, as pointed out by Gouin \emph{et al.} in \cite{Gouinmuller2012}, M\"uller proved that the definition \eqref{muller} is compatible with the entropy principle only if the density is a constant function $ \rho(T)=\rho_0 $. On assuming constant fluid density, no buoyancy-driven convective instabilities are allowed. However, according to experimental observations, fluids expand when heated and a theoretical assumption such as the very widely employed Oberbeck-Boussinesq approximation (see \cite{oberbeck1879,boussinesq1903}) - which consists in setting constant the density of the fluid in all terms of the governing equations except in the body force term due to gravity - is actually reasonable. Therefore, in order to account for the experimental validity of the problem and its thermodynamic consistency, Gouin \emph{et. al} in \cite{Gouinmuller2012} defined a new class of fluids, the "quasi-thermal-incompressible fluids", modifying the constitutive equations \eqref{muller}: \emph{a quasi-thermal-incompressible fluid is a medium for which the only equation independent of the pressure $p$ among all the constitutive equations is the fluid density}. For such class of fluids, the constitutive equations \eqref{muller} become:
\begin{equation}\label{muller2}
	\rho=\rho(T), \quad \varepsilon=\varepsilon(p,T)
\end{equation}
Using the above definition, the authors proved that a quasi-thermal-incompressible fluid tends to be perfectly incompressible, in the sense of M\"uller, when the following estimate for the pressure holds:
\begin{equation}\label{estimate}
	p \ll \dfrac{c_p}{\vert V' \vert}=\dfrac{\rho^2 c_p}{\vert \rho' \vert}
\end{equation}
where $c_p$ is the specific heat capacity at constant pressure. In convection problems, there are no sharp temperature variations and, since the temperature variation usually does not exceed $10 K$, the density variation is of $1 \% $, see \cite{chandr}, therefore the Oberbeck-Boussinesq approximation is coherently employed. When one does not expect large differences in temperature, one may assume the fluid density in the body force term has a linear dependence on temperature:  
\begin{equation}\label{OB}
	\rho(T)=\rho_0 [1-\alpha(T-T_0)]
\end{equation}
where $\rho_0$ is the fluid density at the reference temperature $T_0$, while $\alpha$ is the thermal expansion coefficient, defined as:
$$\alpha=\dfrac{V_T}{V}$$
$V$ being the specific volume and $V_T$ the partial derivative of $V$ with respect to temperature $T$. When \eqref{OB} is assumed, the estimate \eqref{estimate} becomes:
\begin{equation}\label{estimate2}
	p \ll p_{cr}=\dfrac{c_p \rho_0}{\alpha} 
\end{equation}
The critical pressure value $p_{cr}$ gives a limit of validity for the Oberbeck-Boussinesq approximation and due to estimate \eqref{estimate2}, Gouin \emph{et. al} concluded that a quasi-thermal-incompressible fluid is experimentally similar to a perfectly incompressible fluid. \\
Later on, with the aim of proposing a more realistic model for fluid dynamics problems, Gouin and Ruggeri in \cite{Gouinrogers2012} introduced the definition of \emph{extended-quasi-thermal-incompressible} fluid by which they modified the Oberbeck-Boussinesq approximation as follows: 
\begin{equation}\label{OB2}
	\rho(p,T)=\rho_0 [1-\alpha(T-T_0)+\beta(p-p_0)]
\end{equation}
where $p_0$ is the reference pressure, while $\beta$ is the compressibility factor defined as
$$\beta=-\dfrac{V_p}{V}$$
with $V_p$ the partial derivative of the volume with respect to the pressure. Moreover, the Authors carried out a detailed analysis of the thermodynamic stability, proving that the compressibility factor has a lower bound, namely:
\begin{equation}\label{estimate3}
  \beta > \beta_{cr}=\dfrac{\alpha^2 TV}{c_p}(>0).
\end{equation}
It is possible to evaluate the order of magnitude of both critical pressure $p_{cr}$ and compressibility factor $\beta_{cr}$, \eqref{estimate2} and \eqref{estimate3} in the case of liquid water (see \cite{HandbookChemyPhys2005}), since: 
$$ T_0=293 \ K, \ p_0=10^5 \ Pa, \ V_0=10^{-3} \ m^3/kg, \ \rho_0=10^3 \ kg/m^3, \ c_p=4.2 \cdot 10^3 \ J/kg \ K, \ \alpha=207 \cdot 10^{-6} /K  $$
they assume the following values:
$$p_{cr}=2 \cdot 10^{10} \ Pa=2 \cdot 10^{5} \ atm \quad \text{and} \quad \beta_{cr}= 3 \cdot 10^{-12} /Pa.$$  

In \cite{PassRogers2014} such extended approximation was employed for the linear instability analysis of the conduction solution for the classical B\'enard problem, and the Authors proved via linear instability analysis the destabilizing effect of a dimensionless parameter $\widehat{\beta}$, proportional to the positive compressibility factor $\beta$, on the onset of convection. \\
To the best of our knowledge, there is a lack of investigations on the onset of convective motions in porous media assuming the definition of extended-quasi-thermal-incompressible fluid. This lack motivated the present paper. In Section \ref{model} we derive the mathematical model describing the onset of convection for the Darcy-B\'enard problem, while in Section \ref{linear} we perform a linear instability analysis of the thermal conduction solution. In Section \ref{disc} we analyse the asymptotic behaviour of the critical Rayleigh-Darcy number $\Ra$ with respect to the dimensionless compressibility factor $\widehat{\beta}$, proving the destabilizing effect of $\widehat{\beta}$ on the onset of convective instabilities. The paper ends with a concluding Section that recaps all the results.

\section{Mathematical Model}\label{model}
Let us consider a reference frame $Oxyz$ with fundamental unit vectors $\{ \textbf{i},\textbf{j},\textbf{k}\}$ ($\textbf{k}$ pointing vertically upwards) and a horizontal layer $L=\R^2 \times [0,d]$ of fluid-saturated porous medium. To derive the governing equations for the seepage velocity $\textbf{v}$, the temperature field $T$ and the pressure field $p$, let us employ the modified Oberbeck-Boussinesq approximation, see \cite{PassRogers2014}: 
\begin{itemize}
    \item the fluid density $\rho$ is constant in all terms of the governing equations (i.e. $\rho=\rho_0$), except in the buoyancy term;
    \item in the body force term, the constitutive law for the fluid density is given by
    \begin{equation}
        \rho(T)=\rho_0 [1-\alpha (T-T_0) + \beta (p-p_0)]
    \end{equation}
    with $\alpha$ and $\beta$ the thermal expansion coefficient and the compressibility factor, respectively, defined as $$\alpha=\dfrac{V_T}{V}, \quad \beta=-\dfrac{V_p}{V}$$
    \item $\nabla \cdot \textbf{v}=0$ and $\textbf{D} : \textbf{D} \approx 0$.
    \end{itemize}
    Therefore, the mathematical model, according to Darcy's law, is the following
\begin{equation}\label{syst1}
    \begin{cases}
    \dfrac{\mu}{K} \textbf{v}= - \nabla p -\rho_0 [1-\alpha (T-T_0) + \beta (p-p_0)] g \textbf{k} \\ \nabla \cdot \textbf{v} = 0 \\ \rho c_V \Bigl( \dfrac{\partial T}{\partial t} + \textbf{v} \cdot \nabla T \Bigr) = \chi \Delta T
    \end{cases}
\end{equation}
where $\mu, K, \chi, c_V$ are fluid viscosity, permeability of the porous body, thermal conductivity and specific heat at constant volume, respectively. \\ 
To system (\ref{syst1}) the boundary conditions are appended, i.e.:
\begin{equation}\label{BC1}
\begin{aligned}
    \textbf{v} \cdot \textbf{k} =0 \quad & \text{on} \ z=0,d \\
    T=T_L  \quad & \text{on} \ z=0 \\
    T=T_U  \quad & \text{on} \ z=d \\
    \nabla p \cdot \textbf{k} + \rho_0 d \beta g \ p=0 \quad & \text{on} \ z=0,d
\end{aligned}
\end{equation}
with $T_L>T_U$, since the layer is heated from below. Assuming the reference temperature $T_0=T_L$, system (\ref{syst1})-(\ref{BC1}) admits the following stationary conduction solution 
\begin{equation}
 \begin{aligned}
 &  \textbf{v}_b=\textbf{0}, \ T_b(z)=T_L-\dfrac{T_L-T_U}{d} z, \\
  & p_b(z)=\dfrac{1}{\beta d} + \Bigl[ \dfrac{1}{\beta}  - \dfrac{\alpha(T_L-T_U)}{\beta^2 \rho_0 g d} \Bigr] (e^{-\rho_0 g \beta z} -1) - \dfrac{\alpha (T_L-T_U)}{\beta d} z .
 \end{aligned}
\end{equation}
Let $(\textbf{u},\theta,\pi)$ be a perturbation to the basic solution, so the equations governing the perturbation fields are
\begin{equation}\label{pert1}
     \begin{cases}
    \dfrac{\mu}{K} \textbf{u} = - \nabla \pi + \rho_0 \alpha g \theta  \textbf{k} -  \rho_0 \beta g \pi  \textbf{k} \\ \nabla \cdot \textbf{u} = 0 \\ \dfrac{\partial \theta}{\partial t} + \textbf{u} \cdot \nabla \theta = \dfrac{T_L-T_U}{d} \textbf{u} \cdot \textbf{k} + k \Delta \theta
    \end{cases}
\end{equation}
where $k=\frac{\chi}{\rho c_V}$ is the thermal diffusivity. Let us introduce the following scales 
$$\pi=P \pi^*, \quad \textbf{u}=U \textbf{u}^*, \quad \theta=T^{\#} \theta^*, \quad t=\tau t^*, \quad x=dx^* $$
where:
\[ P=\dfrac{\mu k}{K},\quad U=\dfrac{k}{d},\quad T^\#=T_L-T_U, \quad \tau=\dfrac{d^2}{k}. \]
Therefore, the corresponding dimensionless system of equations, omitting all the stars, is the following:

\begin{equation}\label{nondimsys}
    \begin{cases}
    \textbf{u}=-\nabla\pi+\Ra \theta \textbf{k}-\widehat{\beta}\pi \textbf{k}\\
    \nabla \cdot \textbf{u}=0\\
    \dfrac{\partial \theta}{\partial t}+\textbf{u}\cdot \nabla \theta=w+\Delta \theta
    \end{cases}
\end{equation}
where $u=\textbf{u} \cdot \textbf{i}$ and $w=\textbf{u} \cdot \textbf{k}$ and
\[ \Ra=\dfrac{\rho_0\alpha g d(T_L-T_U)K}{\mu k}, \quad \widehat{\beta}=\rho_0d g \beta  \]
are the Rayleigh-Darcy number and the dimensionless compressibility factor, respectively. \\ \\ 
To system \eqref{nondimsys} we add the following boundary conditions

\begin{equation}\label{nondimbc}
    w=\theta=\nabla\pi\cdot \textbf{k}+\widehat{\beta}\pi=0\qquad \text{on}\; z=0,1
\end{equation}
and initial conditions
\begin{equation}\label{incond}
    {\bf u}({\bf x},0)={\bf u}_0({\bf x}), \quad \pi({\bf x},0)=\pi_0({\bf x}), \quad \theta({\bf x},0)=\theta_0({\bf x}).
\end{equation}
Accounting for \eqref{nondimsys}$_2$, taking the divergence of \eqref{nondimsys}$_1$, system \eqref{nondimsys} becomes:

\begin{equation}\label{sys1}
 \begin{cases}
    \Delta \pi+\widehat{\beta}\dfrac{\partial \pi}{\partial z}=\Ra \dfrac{\partial \theta}{\partial z}\\
    \textbf{u}=-\nabla\pi+\Ra \theta \textbf{k}-\widehat{\beta}\pi \textbf{k}\\[2mm]
    \dfrac{\partial \theta}{\partial t}+\textbf{u}\cdot \nabla \theta=w+\Delta \theta 
 \end{cases}
\end{equation}
\begin{rem}
In the sequel, we will focus on bi-dimensional perturbations in the plane $(x,z)$ and assume the perturbations fields $\pi, \textbf{u}, \theta$ to be periodic functions in the horizontal direction $x$ with period $\frac{2 \pi}{a_x}$, $a_x$ being the wavenumber. Without loss of generality, in the sequel we will assume that the wavelength is 1, so $\frac{2 \pi}{a_x}=1$ (see \cite{PassRogers2014, corli2019}) and we will consider the periodicity cell $V$ given by: $$V=[0,1] \times [0,1].$$   
Moreover, with $\|\cdot\|$ and $\scal{\cdot}{\cdot}$ we will denote norm and scalar product on $L^2(V)$, respectively.  
\end{rem}

\section{Linear instability analysis}\label{linear}
To perform the linear instability analysis of the basic solution, let us consider the linear version of \eqref{sys1}:

\begin{equation}\label{sys1lin}
 \begin{cases}
    \Delta \pi+\widehat{\beta}\dfrac{\partial \pi}{\partial z}=\Ra \dfrac{\partial \theta}{\partial z}\\
    \textbf{u}=-\nabla\pi+\Ra \theta \textbf{k}-\widehat{\beta}\pi \textbf{k}\\[2mm]
    \dfrac{\partial \theta}{\partial t}=w+\Delta \theta 
 \end{cases}
\end{equation}
together with boundary conditions:

\begin{equation}\label{BC}
    w=\theta=0\quad \text{and}\quad \dfrac{\partial \pi}{\partial z}=-\widehat{\beta}\pi\qquad \quad \text{on}\; z=0,1
\end{equation}
By virtue of the Robin boundary condition \eqref{BC} on the pressure, it is possible to choose:

\begin{equation}
    \pi=e^{-\widehat{\beta}z}\Pi(x,z,t).
\end{equation}
Therefore equation \eqref{sys1lin}$_1$ becomes:

\begin{equation}
 \Delta \Pi-\widehat{\beta}\dfrac{\partial \Pi}{\partial z}=\Ra e^{\widehat{\beta}z} \dfrac{\partial \theta}{\partial z}
\end{equation}
and the Robin boundary conditions $\dfrac{\partial \pi}{\partial z}=-\widehat{\beta}\pi$ becomes the Neumann condition given by: 

\begin{equation}
    \dfrac{\partial \Pi}{\partial z}=0\qquad z=0,1
\end{equation}
Introducing the stream function $\Phi$ such that

\begin{equation}\label{SF}
    u=-\dfrac{\partial \Phi}{\partial z},\quad w=\dfrac{\partial \Phi}{\partial x}
\end{equation}
and considering the curl of \eqref{sys1lin}$_2$ projected on the $y$-axis, one obtains:

\begin{equation}
    \Delta \Phi=\Ra \dfrac{\partial \theta}{\partial x}-\widehat{\beta}e^{-\widehat{\beta}z}\dfrac{\partial \Pi}{\partial x}
\end{equation}
Hence, to perform the linear instability analysis of the conduction solution, we consider the following system:

\begin{equation}\label{system2}
    \begin{cases}
    \Delta \Pi-\widehat{\beta}\dfrac{\partial \Pi}{\partial z}=\Ra e^{\widehat{\beta}z} \dfrac{\partial \theta}{\partial z}\\[2mm]
  \Delta \Phi=\Ra \dfrac{\partial \theta}{\partial x}-\widehat{\beta}e^{-\widehat{\beta}z}\dfrac{\partial \Pi}{\partial x}\\[2mm]
      \dfrac{\partial \theta}{\partial t}=\dfrac{\partial \Phi}{\partial x}+\Delta \theta
    \end{cases}
\end{equation}
to which we add the boundary conditions: 
\begin{equation}\label{BCtpf}
  \theta=\dfrac{\partial \Pi}{\partial z}=\Delta \Phi=0 \qquad \text{on} \ z=0,1.
\end{equation}
By virtue of \eqref{BCtpf}, since system \eqref{system2} is linear, we assume normal mode solutions:
\begin{equation}\label{series}
    \begin{aligned}
        \theta(x,z,t)&=\sum_{m,n=0}^{\infty}[A^1_{mn}(t)\cos(2\pi mx)\sin(\pi nz)+A^2_{mn}(t)\sin(2\pi mx)\sin(\pi nz)],\\
        \Pi(x,z,t)&=\sum_{m,n=0}^{\infty}[B^1_{mn}(t)\cos(2\pi mx)\cos(\pi nz)+B^2_{mn}(t)\sin(2\pi mx)\cos(\pi nz)],\\
        \Delta \Phi(x,z,t)&=\sum_{m,n=0}^{\infty}[C^1_{mn}(t)\cos(2\pi mx)\sin(\pi nz)+C^2_{mn}(t)\sin(2\pi mx)\sin(\pi nz)].
    \end{aligned}
\end{equation}
In order to get zero mean value on $V$, we assume $(m,n)\in\N\times \N_0$. Applying the laplacian operator to \eqref{system2}$_3$ and by virtue of \eqref{series}, one obtains:

\begin{equation}\label{fourier}
    \!\!\!\! \begin{cases} \displaystyle
    \sum_{m,n}  [B^1_{mn}\cos(2\pi mx)+B^2_{mn}\sin(2\pi mx)] [-\alpha_{m n} \cos(n \pi z) + \widehat{\beta} n \pi \sin(n \pi z)] = \\ \displaystyle \qquad \qquad \Ra \sum_{m,n}  n \pi [A^1_{mn}\cos(2\pi mx)+A^2_{mn}\sin(2\pi mx)] e^{\widehat{\beta} z} \cos(n \pi z) \\[6mm]
    \displaystyle
    \sum_{m,n}  [C^1_{mn}\cos(2\pi mx)+C^2_{mn}\sin(2\pi mx)] \sin(n \pi z) = \\ \displaystyle \qquad \qquad  \sum_{m,n=0}  2 \pi m \Bigl\{ \Ra  [-A^1_{mn}\sin(2\pi mx)+A^2_{mn}\cos(2\pi mx)] \sin(\pi nz) \\ \displaystyle \qquad \qquad - \widehat{\beta} e^{-\widehat{\beta} z} [-B^1_{mn}\sin(2\pi mx)+B^2_{mn}\cos(2\pi mx)] \cos(n \pi z) \Bigr\} \\[6mm] 
    
    \displaystyle
    \sum_{m,n} -  \alpha_{m n}  [(\dot{A}^1_{mn} + \alpha_{mn} A^1_{mn}) \cos(2\pi mx)+  (\dot{A}^2_{mn} + \alpha_{mn} A^2_{mn})   \sin(2\pi mx)]  \sin(n \pi z) = \\ \displaystyle \qquad \qquad  \sum_{m,n} 2 \pi m [-C^1_{mn}\sin(2\pi mx)+C^2_{mn}\cos(2\pi mx)] \sin(n \pi z)
    \end{cases} 
\end{equation}
where $\alpha_{m n}=(2 \pi m)^2+(\pi n)^2$ and $\dot{A}^i_{mn}=\dfrac{dA^i_{mn}}{dt}$. From \eqref{fourier}$_3$, it immediately follows that 

\begin{equation}\label{C}
    \begin{aligned}
    C^1_{mn} &=  \dfrac{\alpha_{mn}}{2 \pi m} (\dot{A}^2_{mn} + \alpha_{mn} A^2_{mn}) \\ C^2_{mn} &= -  \dfrac{\alpha_{mn}}{2 \pi m} (\dot{A}^1_{mn} + \alpha_{mn} A^1_{mn})
    \end{aligned}
\end{equation}
Let us multiply \eqref{fourier}$_1$ by $\cos(k \pi  z)$ and integrate with respect to $z \in (0,1)$, therefore we get:

\begin{equation}
    \mkern-18mu \mkern-36mu \begin{aligned}
    \displaystyle 
    & \sum_{m,n}   [B^1_{mn}\cos(2\pi mx) \!+\! B^2_{mn}\sin(2\pi mx)] \left[\!-\! \alpha_{m n} \!\! \int_0^1  \!\! \cos(n \pi z)\cos(k\pi z) dz \!+\! \widehat{\beta} n \pi \!\! \int_0^1 \!\! \sin(n \pi z)\cos(k\pi z) dz\right] \!=\! \\ & \displaystyle \qquad \qquad  \sum_{m,n} \Ra \  n \pi [A^1_{mn}\cos(2\pi mx)+A^2_{mn}\sin(2\pi mx)] \int_0^1 e^{\widehat{\beta} z} \cos(n \pi z) \cos(k \pi z) dz
    \end{aligned}
\end{equation}
namely:

\begin{equation}\label{30}
    \begin{aligned}
    \displaystyle
    & \sum_{m,n}  [B^1_{mn}\cos(2\pi mx)+B^2_{mn}\sin(2\pi mx)] \left[-\dfrac{1}{2}\alpha_{mk}\delta_{nk}+\widehat{\beta} F_{nk} \right] = \\ & \displaystyle \qquad \qquad  \Ra \sum_{m,n}  \   [A^1_{mn}\cos(2\pi mx)+A^2_{mn}\sin(2\pi mx)] \dfrac{\widehat{\beta}}{2} \mathcal{L}_{nk}(\widehat{\beta})
    \end{aligned}
\end{equation}
whit:

\begin{equation}
\begin{aligned}
    F_{nk} &= \begin{cases} 0 \quad & \text{if} \ n=k \\ \dfrac{n^2((-1)^{n+k}-1)}{(k-n)(k+n)} \quad & \text{if} \ n \ne k \end{cases} \\ \mathcal{L}_{nk}(\widehat{\beta}) &= n\pi (e^{\widehat{\beta}}(-1)^{k+n}-1)\left(\dfrac{1}{\pi^2(k+n)^2+\widehat{\beta}^2}+\dfrac{1}{\pi^2(k-n)^2+\widehat{\beta}^2}\right)
\end{aligned}
\end{equation}

Setting

\begin{equation}
   \mathcal{D}^m_{nk}(\widehat{\beta})=\delta_{nk}+\widehat{\beta}\begin{cases}
   \dfrac{-2n}{\alpha_{mk}}\left(\dfrac{1}{n-k}+\dfrac{1}{n+k}\right)&\quad \text{if}\quad n+k\;\;\text{odd}\\
   0&\quad \text{if}\quad n+k\;\;\text{even}
   \end{cases}
\end{equation}
by virtue of linearity, from \eqref{30} one obtains 
\begin{equation}\label{32}
  \begin{aligned}
    \displaystyle
     \sum_{n=0}^{\infty}  B^i_{mn} \mathcal{D}^m_{nk}(\widehat{\beta}) =  -\dfrac{\Ra\widehat{\beta}}{\alpha_{mk}} \sum_{n=0}^{\infty}  \   A^i_{mn} \mathcal{L}_{nk}(\widehat{\beta}), \qquad i=1,2.
    \end{aligned}
\end{equation}
Let us remark that the $N \times N$ matrix $\mathcal{D}^m$ is invertible since it is strictly diagonally dominant for small $\widehat{\beta}$ (see \cite{matrices}). Moreover, through a fixed point argument, an estimate on the compressibility factor $\widehat{\beta}$ guaranteeing the invertibility of the matrix $\mathcal{D}^m$ for all $N \in \N$ is obtained. The following theorem holds. 
\\
\begin{thm}\label{thmD}
If 
\begin{equation}\label{bbbb}
    \b<\dfrac{\pi^2}{2 c}  
\end{equation}
with $ c= \dfrac{1}{8} [2 \pi \coth({2 \pi})+1] $, then the matrix $\mathcal{D}^m$ is invertible for all $N\in\N$. 
\end{thm}

\begin{proof}
Let us consider the basis functions:
\begin{equation}\label{eignefunct}
    \phi^i_{mn}(x,z)=\begin{cases}
    \cos(2\pi mx)\cos(n \pi z)& \text{if}\;\; i=1\\
    \sin(2\pi mx)\cos(n \pi z)& \text{if}\;\; i=2
    \end{cases}
\end{equation}
which are the eigenfunctions of the Laplace operator:
\begin{equation}
    \Delta \phi^i_{mn}=-\alpha_{mn}\phi^i_{mn},
\end{equation}
 $\alpha_{mn}=4\pi^2m^2+\pi^2n^2$ being the eigenvalues. Since:
 
 \begin{equation}
     b_{mn}:=\|\phi^i_{mn}\|^2=\begin{cases}
     \dfrac{1}{2}&n=0 \\[2mm]
     \dfrac{1}{4}&\text{otherwise}
     \end{cases}
 \end{equation}
 defining $\gamma_{mn}=\sqrt{\alpha_{mn}b_{mn}}$,
 the following normalization can be introduced:
 
 \begin{equation}\label{normalizedbasis}
     \psi^i_{mn}=\dfrac{\phi_{mn}^i}{\gamma_{mn}}.
 \end{equation}
Equation \eqref{system2}$_1$ can be written in terms of \eqref{normalizedbasis} as:

\begin{equation}\label{38}
    -\sum_{i=1,2 \atop m,n} B^i_{mn}(-\Delta \psi^i_{mn})-\widehat{\beta} \sum_{i=1,2 \atop m,n}B^i_{mn}\dfrac{\partial \psi^i_{mn}}{\partial z}= \sum_{i=1,2 \atop m,n} e^{\b z}\Ra A^i_{mn}\dfrac{\partial \psi^i_{mn}}{\partial z}.
\end{equation}
If we multiply \eqref{38} by $\psi^j_{lr}$ and integrate on $V$ we obtain:

\begin{equation}\label{39}
      -\sum_{i=1,2 \atop m,n} B^i_{mn}\scal{\nabla \psi^i_{mn}}{\nabla \psi^j_{lr}}-\widehat{\beta} \sum_{i=1,2 \atop m,n}B^i_{mn}\scal{\dfrac{\partial \psi^i_{mn}}{\partial z}}{\psi^j_{lr}}= \sum_{i=1,2 \atop m,n} F^{i,j}_{mnlr}
\end{equation}
where:

\begin{equation}
    F^{i,j}_{mnlr}=e^{\b z}\Ra A^i_{mn}\scal{\dfrac{\partial \psi^i_{mn}}{\partial z}}{\psi^j_{lr}}
\end{equation}
From \eqref{eignefunct} and \eqref{normalizedbasis}, it follows that:
 
 \begin{equation}
     \scal{\dfrac{\partial \psi^i_{mn}}{\partial z}}{\psi^j_{lr}}=\dfrac{1}{\gamma_{mn}\gamma_{lr}}\scal{\dfrac{\partial \phi^i_{mn}}{\partial z}}{\phi^j_{lr}}=-\dfrac{\delta_{ij}\delta_{ml}}{\gamma_{mn}\,\gamma_{lr}}\dfrac{n}{2}\begin{cases}
     \frac{1}{n+r}+\frac{1}{n-r} & \text{se} \;\; n+r\geq 1 \;\;\text{odd}\\[2mm]
     0& \text{otherwise}
     \end{cases}
 \end{equation}
and equation \eqref{39} becomes:
 
 \begin{equation}\label{42}
     -B^j_{lr}+\b \sum_{i \,|\, m,n \atop n+r\geq 1\;\;\text{odd}} B^i_{mn}\dfrac{n}{2}\dfrac{\delta_{ij}\delta_{ml}}{\gamma_{mn}\gamma_{lr}}\left(\frac{1}{n+r}+\frac{1}{n-r}\right)-\sum_{i | m,n } F^{i,j}_{mnlr}=0
 \end{equation}
 Now, let us introduce the following continuous functions:
 
 \begin{equation}
 \begin{split}
    \mathcal{P}&:B\in\R^N\longmapsto  -B^j_{lr}+\b \sum_{i | m,n \atop n+r\geq 1\;\;\text{odd}} B^i_{mn}\dfrac{n}{2}\dfrac{\delta_{ij}\delta_{ml}}{\gamma_{mn}\gamma_{lr}}\left(\frac{1}{n+r}+\frac{1}{n-r}\right)-\sum_{i | m,n } F^{i,j}_{mnlr}\in\R^N   \\[2mm]
    \mathcal{G}&:B\in\R^N\longmapsto  \b \sum_{i | m,n \atop n+r\geq 1\;\;\text{odd}} B^i_{mn}\dfrac{n}{2}\dfrac{\delta_{ij}\delta_{ml}}{\gamma_{mn}\gamma_{lr}}\left(\frac{1}{n+r}+\frac{1}{n-r}\right)-\sum_{i | m,n } F^{i,j}_{mnlr}\in\R^N
 \end{split}
 \end{equation}
so the algebraic system \eqref{42} - equivalent to system \eqref{32} - can be written as $\mathcal{P}(B)=0$. Let us observe that the invertibility of $\mathcal{D}^m$ is equivalent to prove that system \eqref{42} admits a nontrivial solution, moreover, $B$ is a solution of \eqref{42} if and only if $B$ is a fixed point of $\mathcal{G}$:

\begin{equation}
    \mathcal{P}(B)=0\quad \Longleftrightarrow\quad \mathcal{G}(B)=B.
\end{equation}
The existence of a fixed point for $\mathcal{G}$ is guaranteed by the Leray–Schauder theorem, provided that:

\begin{equation}\label{45}
   \{ B\in\R^N\;|\;B=\lambda \mathcal{G}(B),\; 0\leq \lambda\leq 1 \}\subset B_R(0)
\end{equation}
$B_R(0)$ being a ball of radius $R>0$ centered in $0$, hence:

\begin{equation}
   \{ B\in\R^N\;|\;B=\lambda \mathcal{G}(B),\; 0\leq \lambda\leq 1 \}^\complement\supset \mathscr{B}:=\{ B\in\R^N\;|\;B=\lambda \mathcal{G}(B),\; \lambda> 1 \}.
\end{equation}
If $B\in\mathscr{B}$, then $\lambda B+(1-\lambda)B=\lambda\mathcal{G}(B)$, i.e. $(1-\lambda )B=\lambda \mathcal{P}(B)$, therefore:

\begin{equation}\label{47}
    \dfrac{1-\lambda}{\lambda}|B|^2=\mathcal{P}(B)\cdot B
\end{equation}
$|\cdot|$ being the standard euclidean norm. Therefore, from \eqref{45} and \eqref{47} we can state that the proof of the existence of a fixed pointy for $\mathcal{G}$ is equivalent to prove that: 

\begin{equation}\label{48}
\mathcal{P}(B)\cdot B=-\sum_{l,r}(B^j_{lr})^2+\dfrac{\b}{2} \sum_{i | m,n,r \atop n+r\geq 1\;\;\text{odd}} B^i_{mn}B^i_{mr}\dfrac{1}{\gamma_{mn}\gamma_{mr}}\left(\frac{n}{n+r}+\frac{n}{n-r}\right)-\sum_{i | m,n,l,r } F^{i,j}_{mnlr}B^i_{lr}
\end{equation}
is negative for $|B|>R$. For notational convenience let us set: 

\begin{equation}
\widetilde{B}^i_{mnr}=\dfrac{B^i_{mn}}{\gamma_{mr}}\qquad \text{and}\qquad     \widetilde{B}^i_{mrn}=\dfrac{B^i_{mr}}{\gamma_{mn}}
\end{equation}
and hence, from \eqref{48} we have:

\begin{equation}\label{50}
  \dfrac{1}{2} \!\!\! \sum_{i | m,n,r \atop n+r\geq 1\;\;\text{odd}} \widetilde{B}^i_{mnr}\widetilde{B}^i_{mrn}\left(\frac{n}{n+r}+\frac{n}{n-r}\right)=:I+J
\end{equation}
Therefore:

\begin{equation}\label{51}
\begin{split}
 I&= \sum_{i | m,n,r \atop n+r\geq 1\;\;\text{odd}} \widetilde{B}^i_{mnr}\widetilde{B}^i_{mrn}\frac{n}{n+r}=\sum_{i | m \atop n\;\text{even}, \; r\;\text{odd}} \widetilde{B}^i_{mnr}\widetilde{B}^i_{mrn}\frac{n}{n+r}+\sum_{i | m \atop n\;\text{odd}, \; r\;\text{even}} \widetilde{B}^i_{mnr}\widetilde{B}^i_{mrn}\frac{n}{n+r} \\[4mm]
 &=\sum_{i | m \atop n\;\text{even}, \; r\;\text{odd}} \widetilde{B}^i_{mnr}\widetilde{B}^i_{mrn}\frac{n}{n+r}+\sum_{i | m \atop n\;\text{even}, \; r\;\text{odd}} \widetilde{B}^i_{mnr}\widetilde{B}^i_{mrn}\frac{r}{n+r}=\sum_{i | m \atop n\;\text{even}, \; r\;\text{odd}} \widetilde{B}^i_{mnr}\widetilde{B}^i_{mrn}
\end{split}
\end{equation}
and similarly with $J$ 
\begin{equation}\label{51-bis}
\begin{split}
 J&= \sum_{i | m,n,r \atop n+r\geq 1\;\;\text{odd}} \widetilde{B}^i_{mnr}\widetilde{B}^i_{mrn}\frac{n}{n-r}=\sum_{i | m \atop n\;\text{even}, \; r\;\text{odd}} \widetilde{B}^i_{mnr}\widetilde{B}^i_{mrn}\frac{n}{n-r}+\sum_{i | m \atop n\;\text{odd}, \; r\;\text{even}} \widetilde{B}^i_{mnr}\widetilde{B}^i_{mrn}\frac{n}{n-r} \\[4mm]
 &=\sum_{i | m \atop n\;\text{even}, \; r\;\text{odd}} \widetilde{B}^i_{mnr}\widetilde{B}^i_{mrn}\frac{n}{n-r}+\sum_{i | m \atop n\;\text{even}, \; r\;\text{odd}} \widetilde{B}^i_{mnr}\widetilde{B}^i_{mrn}\frac{r}{n-r}=\sum_{i | m \atop n\;\text{even}, \; r\;\text{odd}} \widetilde{B}^i_{mnr}\widetilde{B}^i_{mrn}
\end{split}
\end{equation}
By virtue of \eqref{51} and \eqref{51-bis}, Cauchy-Schwarz and Young inequalities and since $\gamma^2_{mn}\geq \dfrac{4\pi^2+n^2\pi^2}{4}$, from \eqref{50} it follows:

\begin{equation}\label{52}
   \begin{split}
       \sum_{i | m \atop n\;\text{even}, \; r\;\text{odd}} \widetilde{B}^i_{mnr}\widetilde{B}^i_{mrn}&=\sum_{i \atop n\;\text{even}, \; r\;\text{odd}} \widetilde{B}^i_{\cdot, nr}\cdot \widetilde{B}^i_{\cdot, rn}\leq \sum_{i \atop n\;\text{even}, \; r\;\text{odd}} |\widetilde{B}^i_{\cdot, nr}|\; |\widetilde{B}^i_{\cdot, rn}|\\[2mm]
       &\leq \dfrac{1}{2}\left[ \sum_{i \atop n\;\text{even}, \; r\;\text{odd}} |\widetilde{B}^i_{\cdot, nr}|^2+\sum_{i \atop n\;\text{even}, \; r\;\text{odd}} |\widetilde{B}^i_{\cdot, rn}|^2 \right]\\[2mm]
       &=\dfrac{1}{2}\left[ \sum_{i\;|\; m \atop n\;\text{even}, \; r\;\text{odd}} \dfrac{|{B}^i_{mn}|^2}{\gamma^2_{mr}}+\sum_{i\;|\; m \atop n\;\text{even}, \; r\;\text{odd}} \dfrac{|{B}^i_{mr}|^2}{\gamma^2_{mn}} \right]\\[2mm]
       &\leq  2\left[ \sum_{i\;|\; m \atop n\;\text{even}} |{B}^i_{mn}|^2\sum_{r\;\text{odd} \atop} \dfrac{1}{4\pi^2+r^2\pi^2} +\sum_{i\;|\; m \atop r\;\text{odd}} |{B}^i_{mr}|^2\sum_{n\;\text{even} \atop} \dfrac{1}{4\pi^2+n^2\pi^2}  \right]\\
       &\leq \dfrac{2}{\pi^2}\left[ \sum_{i\;|\; m \atop n\;\text{even}} |{B}^i_{mn}|^2 +\sum_{i\;|\; m \atop r\;\text{odd}} |{B}^i_{mr}|^2  \right]\sum_{n \atop} \dfrac{1}{4+n^2}\\
       &=\dfrac{2}{\pi^2}|B|^2\sum_{n \atop} \dfrac{1}{4+n^2} = \dfrac{2}{\pi^2}|B|^2c
   \end{split} 
\end{equation}
where $c=\dfrac{1}{8} [2 \pi \coth({2 \pi})+1]$. Finally, from \eqref{48} and \eqref{52} one gets:

\begin{equation}
    \mathcal{P}(B)\cdot B\leq -|B|^2+\b c\dfrac{2}{\pi^2}|B|^2+K|B|
\end{equation}
with $K=|F|$. Therefore, for $|B|>R := K/(1-\widehat{\beta} c 2 \pi^{-2})$ and if

\begin{equation}
\b<\dfrac{\pi^2}{2c}    
\end{equation}
it follows $ \mathcal{P}(B)\cdot B<0$. 

\end{proof}

Solving system \eqref{32}, we get component-wise the same relation for the coefficients $B^1_{mj}$ and $B^2_{mj}$, i.e. for $i=1,2$:

\begin{equation}\label{B}
B^i_{mj}=-\dfrac{\Ra  \widehat{\beta}}{\alpha_{mk}} \sum_{n,k=0}^{\infty}A^i_{mn}\mathcal{L}_{nk}(\widehat{\beta}) [\mathcal{D}^m(\widehat{\beta})]^{-1}_{kj}.
\end{equation}
Now, let us substitute \eqref{C} and \eqref{B} in \eqref{fourier}$_2$, obtaining:

\begin{equation}\label{coeffA}
\!\!\!\!\!
\begin{split}
  \displaystyle & \sum_{m,j=0}^{\infty} [ \dfrac{\alpha_{mj}}{2 \pi m} (\dot{A}^2_{mj} + \alpha_{mj} A^2_{mj}) \cos(2\pi mx)-  \dfrac{\alpha_{mj}}{2 \pi m} (\dot{A}^1_{mj} + \alpha_{mj} A^1_{mj}) \sin(2\pi mx)] \sin(j \pi z) = \\ & \qquad \quad  \sum_{m,j=0}^{\infty}  2 \pi m \Bigl\{ \Ra  [-A^1_{mj}\sin(2\pi mx)+A^2_{mj}\cos(2\pi mx)] \sin(\pi jz) \\ \displaystyle & \qquad \qquad  - \widehat{\beta} e^{-\widehat{\beta} z} \Bigl[ \dfrac{\Ra  \widehat{\beta}}{\alpha_{mj}} \sum_{n,k=0}^{\infty}A^1_{mn}\mathcal{L}_{nk}(\widehat{\beta}) [\mathcal{D}^m(\widehat{\beta})]^{-1}_{kj} \sin(2\pi mx) \\ & \qquad \qquad \qquad \quad \quad -\dfrac{\Ra  \widehat{\beta}}{\alpha_{mj}} \sum_{n,k=0}^{\infty}A^2_{mn}\mathcal{L}_{nk}(\widehat{\beta}) [\mathcal{D}^m(\widehat{\beta})]^{-1}_{kj} \cos(2\pi mx)\Bigr] \cos(j \pi z) \Bigr\}    
\end{split}
\end{equation}
Let us multiply \eqref{coeffA} by $\sin(h \pi  z)$ and integrate with respect to $z \in (0,1)$, therefore we get:

\begin{equation}\label{galerkina}
\mkern-36mu \mkern-36mu \begin{split}
  \displaystyle & \sum_{m,j=0}^{\infty} \!\! \left[ \dfrac{\alpha_{mj}}{2 \pi m} (\dot{A}^2_{mj} + \alpha_{mj} A^2_{mj}) \cos(2\pi mx) \!-\!  \dfrac{\alpha_{mj}}{2 \pi m} (\dot{A}^1_{mj} + \alpha_{mj} A^1_{mj}) \sin(2\pi mx) \right] \!\! \int_0^1 \!\! \sin(j \pi z)\sin(h\pi z)dz \!=\! \\ & \qquad \qquad  \sum_{m,j=0}^{\infty}  2 \pi m \Bigl\{ \Ra  [-A^1_{mj}\sin(2\pi mx)+A^2_{mj}\cos(2\pi mx)] \int_0^1\sin(\pi jz)\sin(h\pi z)dz \\ \displaystyle & \qquad  - \widehat{\beta}  \Bigl[\dfrac{\Ra  \widehat{\beta}}{\alpha_{mj}} \sum_{n,k=0}^{\infty}A^1_{mn}\mathcal{L}_{nk}(\widehat{\beta})[\mathcal{D}^m(\widehat{\beta})]^{-1}_{kj} \sin(2\pi mx)\\
  &\qquad\qquad-\dfrac{\Ra  \widehat{\beta}}{\alpha_{mj}} \sum_{n,k=0}^{\infty}A^2_{mn}\mathcal{L}_{nk}(\widehat{\beta})[\mathcal{D}^m(\widehat{\beta})]^{-1}_{kj} \cos(2\pi mx)\Bigr] \int_0^1 e^{-\widehat{\beta} z}\cos(j \pi z)\sin(h\pi z)dz \Bigr\}    
\end{split}  
\end{equation}
Hence:
\begin{equation}\label{galerkina2}
\begin{split}
  \displaystyle & \sum_{m=0}^{\infty}  \left[ \dfrac{\alpha_{mh}}{2 \pi m} (\dot{A}^2_{mh} + \alpha_{mh} A^2_{mh}) \cos(2\pi mx)-  \dfrac{\alpha_{mh}}{2 \pi m} (\dot{A}^1_{mh} + \alpha_{mh} A^1_{mh}) \sin(2\pi mx) \right] = \\ & \qquad \qquad  \sum_{m=0}^{\infty} 2 \pi m \Ra  [-A^1_{mh}\sin(2\pi mx)+A^2_{mh}\cos(2\pi mx)] \\ \displaystyle & \qquad \qquad - \sum_{m=0}^{\infty} \widehat{\beta}^2 \Ra  2 \pi m \sum_{j,n,k=0}^{\infty} \dfrac{1}{\alpha_{mj}} A^1_{mn}\mathcal{L}_{nk}(\widehat{\beta})[\mathcal{D}^m(\widehat{\beta})]^{-1}_{kj} \mathcal{N}_{jh}(\widehat{\beta}) \sin(2\pi mx)   \\  \displaystyle & \qquad \qquad + \sum_{m=0}^{\infty} \widehat{\beta}^2 \Ra  2 \pi m \sum_{j,n,k=0}^{\infty} \dfrac{1}{\alpha_{mj}} A^2_{mn}\mathcal{L}_{nk}(\widehat{\beta}) [\mathcal{D}^m(\widehat{\beta})]^{-1}_{kj} \mathcal{N}_{jh}(\widehat{\beta}) \cos(2\pi mx) 
\end{split}
\end{equation}
with

\begin{equation}
    \mathcal{N}_{jh}(\widehat{\beta})=\pi (1-e^{-\widehat{\beta}}(-1)^{h+j})\left(\dfrac{h+j}{\pi^2(h+j)^2+\widehat{\beta}^2}+\dfrac{h-j}{\pi^2(h-j)^2+\widehat{\beta}^2}\right)
\end{equation}
\\ \\
By the linear independence of the sinus and cosinus functions with respect to the variable $x$, we get, for $i=1,2$:

\begin{equation}\label{eqa}
    \dfrac{\alpha_{mh}}{2\pi m} (\dot{A}^i_{mh} + \alpha_{mh} A^i_{mh})= 2 \pi m \Ra A^i_{mh} + \widehat{\beta}^2 \Ra 2 \pi m \sum_{j,n,k=0}^{\infty} \dfrac{1}{\alpha_{mj}} A^i_{mn}\mathcal{L}_{nk}(\widehat{\beta}) [\mathcal{D}^m(\widehat{\beta})]^{-1}_{kj} \mathcal{N}_{jh}(\widehat{\beta}).
\end{equation}
Equations \eqref{eqa} are first order ODEs with respect to time $t$. To get a unique solution, system \eqref{eqa} decouples and let $A^i_{mh}$ be the only non-vanishing coefficient, which satisfies the following first-order ordinary differential equation:
\begin{equation}\label{eqa2}
  \dot{A}^i_{mh} + \alpha_{mh} {A}^i_{mh} = \dfrac{4\pi^2 m^2}{\alpha_{mh}} \Ra \  {A}^i_{mh} + \widehat{\beta}^2 \Ra  \dfrac{4\pi^2 m^2}{\alpha_{mh}} {A}^i_{mh} \sum_{j,k=0}^\infty \dfrac{1}{\alpha_{mj}}\mathcal{L}_{hk}(\widehat{\beta})[\mathcal{D}^m(\widehat{\beta})]^{-1}_{kj} \mathcal{N}_{jh}(\widehat{\beta}) 
\end{equation}
together with the initial conditions on $A^i_{mh}$ that can be derived from \eqref{incond}$_3$ and \eqref{series}$_1$. Setting

\begin{equation}
    \mathcal{G}_{mh}(\widehat{\beta})=\dfrac{4\pi^2 m^2}{\alpha_{mh}}\sum_{j,k=0}^\infty \dfrac{1}{\alpha_{mj}}\mathcal{L}_{hk}(\widehat{\beta})[\mathcal{D}^m(\widehat{\beta})]^{-1}_{kj} \mathcal{N}_{jh}(\widehat{\beta})
\end{equation}
\eqref{eqa2} is equivalent to

\begin{equation}\label{ode}
    \dot{A}^i_{mh} + A^i_{mh} \Bigl[ \alpha_{mh} -  \Ra \dfrac{4\pi^2 m^2}{\alpha_{mh}}  - \widehat{\beta}^2 \Ra \ \mathcal{G}_{mh}(\widehat{\beta}) \Bigr] =0
\end{equation}
whose solution can be easily computed to be:
\begin{equation}
    A^i_{mh}(t) = \gamma e^{\left( - \alpha_{mh} +  \Ra \frac{4\pi^2 m^2}{\alpha_{mh}}  + \widehat{\beta}^2 \Ra \ \mathcal{G}_{mh}(\widehat{\beta}) \right)t}
\end{equation}
$\gamma$ being a constant depending on the initial conditions. We obtain that the perturbation fields \eqref{series} have an exponential dependence on time $t$, so let  us define the generalized eigenvalue $\sigma_{mh}$:
\begin{equation}\label{eigenvalue}
    \sigma_{mh} = - \alpha_{mh} +  \Ra \dfrac{4\pi^2 m^2}{\alpha_{mh}}  + \widehat{\beta}^2 \Ra \ \mathcal{G}_{mh}(\widehat{\beta})
\end{equation}

\section{Results and discussion}\label{disc} 
\begin{rem}
Let us first underline that the eigenvalues \eqref{eigenvalue} are real $\forall \ m,h$. Therefore, the strong principle of exchange of stabilities holds and convection can arise only via stationary motions.  
\end{rem}
\begin{rem}
In the limit case $\widehat{\beta} \rightarrow 0$ (i.e. according to the classical Oberbeck-Boussinesq approximation), \eqref{eigenvalue} becomes
\begin{equation}\label{eigenvalue0}
    \sigma_{mh}= - \alpha_{mh} +  \Ra \dfrac{4\pi^2 m^2}{\alpha_{mh}} 
\end{equation}
so, requiring the eigenvalue $\sigma_{mh}$ to be positive, we get
\begin{equation}\label{ra}
    - \alpha_{mh} +  \Ra \dfrac{4\pi^2 m^2}{\alpha_{mh}} >0.
\end{equation}
Therefore, convection arises if the Rayleigh-Darcy number is greater than the critical value 
\begin{equation}\label{min0}
    \Ra_c =\min_{m,h} \dfrac{[(2 \pi m)^2 + ( h \pi)^2]^2}{4\pi^2 m^2}
\end{equation}
The minimum \eqref{min0} is obtained for $h=1$ and $m^*= \dfrac{1}{2}$, so the classical result is recovered, i.e. the critical wavenumber is $(2 \pi m^*)^2=\pi^2$, while the critical Rayleigh-Darcy number is:
\begin{equation}
    \Ra_c =4 \pi^2
\end{equation}
\end{rem}
According to \eqref{eigenvalue}, the marginal instability threshold is given setting $\sigma_{mn}=0$, i.e. 
\begin{equation}
    \Ra \Bigl( \dfrac{4\pi^2 m^2}{\alpha_{mh}}  + \widehat{\beta}^2  \mathcal{G}_{mh}(\widehat{\beta}) \Bigr) - \alpha_{mh} =0
\end{equation}
so, when the horizontal layer of porous medium is saturated by an extended-quasi-thermal-incompressible fluid, the critical Rayleigh-Darcy number for the onset of convection is given by:
\begin{equation}\label{min}
    \Ra_L=\inf_{m,h} \dfrac{\alpha_{mh}^2}{4\pi^2 m^2  + \widehat{\beta}^2  \alpha_{mh} \mathcal{G}_{mh}(\widehat{\beta}) }.
\end{equation}
In order to analyse the influence of the dimensionless compressibility factor $\b$ on the onset of convection, we numerically solved \eqref{min} for quoted values of $\widehat{\beta}$, under the restriction \eqref{bbbb} found in Theorem \ref{thmD}. \\ 
We found that the function $\mathcal{G}_{mh}$ is always positive and the dimensionless compressibility factor $\widehat{\beta}$ has a \emph{destabilizing} effect on the onset of convective flows: the behaviour of the critical Rayleigh-Darcy number with respect to $\hat{\beta}$ is decreasing (see Figures \ref{a} -- \ref{b}) and
\begin{equation}
     \Ra_L < \Ra_c \qquad \forall \ \widehat{\beta}>0.
\end{equation}

\begin{figure}[h!]
    \centering
    \includegraphics[scale=0.7]{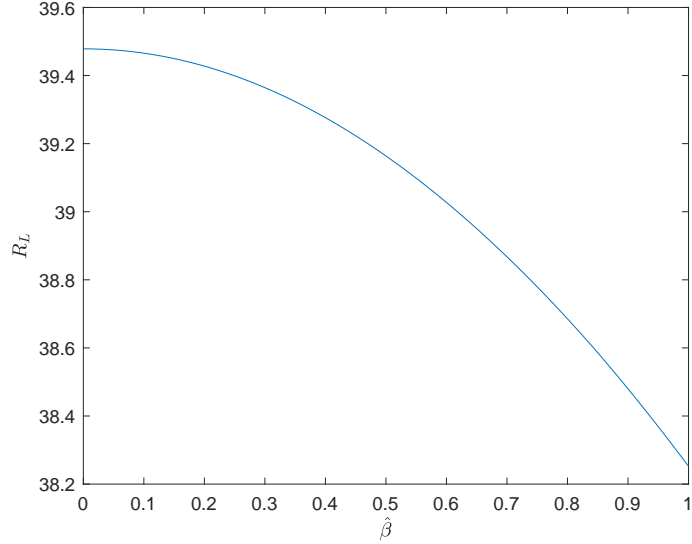}
    \caption{Critical Rayleigh-Darcy number $\Ra_L$ as function of the compressibility factor $\hat{\beta}$.}
    \label{a}
\end{figure}

\begin{figure}[h!]
    \centering
    \includegraphics[scale=0.7]{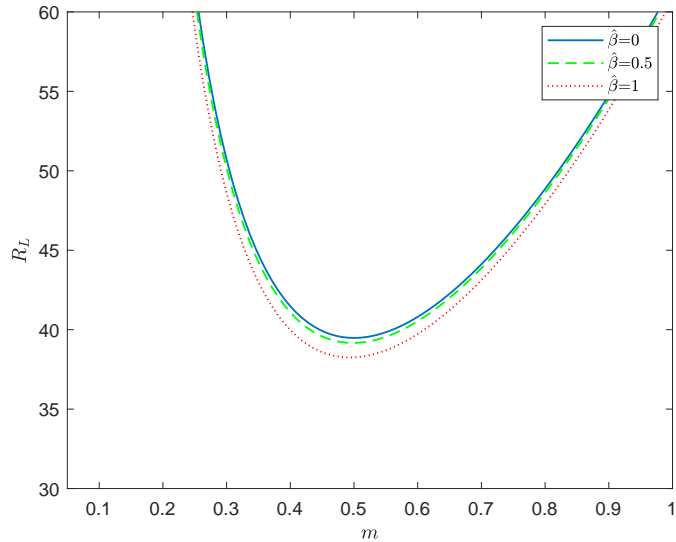}
    \caption{Neutral curves for quoted values of the compressibility factor $\hat{\beta}$.}
    \label{b}
\end{figure}

\section{Conclusions}
To the best of our knowledge, in this paper the Darcy-B\'enard problem for an extended-quasi-thermal-incompressible fluid has been studied for the first time. We determined the instability threshold for the onset of convection via linear instability analysis of the conduction solution: through a closed algebraic form, we showed that the critical Rayleigh-Darcy number depends on the dimensionless compressibility factor $\b$ and we rigorously proved that $\b$ has a destabilizing effect. Moreover, in the limit case $\widehat{\beta} \rightarrow 0$ (i.e. according to the classical Oberbeck-Boussinesq approximation), the critical threshold for the Darcy-B\'enard problem $4 \pi^2$ is recovered. 

\textbf{Acknowledgements.} This paper has been performed under the auspices of the GNFM of INdAM.

\bibliographystyle{unsrt}
\bibliography{main}

\end{document}